\newif\ifshort  \newif\iflong
     \shorttrue \longfalse

\documentclass[orivec]{llncs}

\usepackage[usenames,dvipsnames]{color}
\usepackage{amssymb}
\usepackage{graphicx}
\usepackage{subfigure}
\usepackage{wrapfig}
\usepackage{xspace}

\newcommand{\Graphs}{{\cal G}}
\newcommand{\Dirs}{{\cal V}}
\newcommand{\dir}{\vec{v}}

\newcommand{\pth}[1]{\ensuremath{\left(#1\right)}}
\newcommand{\R}{{\mathbb{R}}}

\newcommand{\floor}[1]{\ensuremath{\left\lfloor#1\right\rfloor}}

\renewcommand{\leq}{\leqslant}
\renewcommand{\geq}{\geqslant}
\renewcommand{\le}{\leqslant}

\definecolor{cyan}{rgb}{0.0, 0.72, 0.92}
\definecolor{pink}{rgb}{1,0,1}
\newcommand{\olivier}[1]{\marginpar{\tiny\color{OliveGreen} O: #1}}
\newcommand{\tamara}[1]{\marginpar{\tiny\color{Fuchsia} T: #1}}
\newcommand{\david}[1]{\marginpar{\tiny\color{Orange} D: #1}}
\newcommand{\marc}[1]{\marginpar{\tiny\color{blue} M: #1}}
\newcommand{\beppe}[1]{\marginpar{\tiny\color{red} B: #1}}
\newcommand{\syl}[1]{{\color{pink}\marginpar{\tiny\color{pink} SL: #1}}}
\newcommand{\sue}[1]{\marginpar{\tiny\color{green} SW: #1}}

\newcommand{\SL}[1]{{\color{pink} #1}}

\renewcommand{\olivier}[1]{}\renewcommand{\tamara}[1]{}\renewcommand{\david}[1]{}\renewcommand{\marc}[1]{}\renewcommand{\beppe}[1]{}\renewcommand{\syl}[1]{}\renewcommand{\sue}[1]{}\renewcommand{\SL}[1]{{#1}}

\iflong
\title{Monotone Simultaneous Embeddings of Paths in $\mathbb{R}^d$ \thanks{Research supported in part by the MIUR
    project AMANDA ``Algorithmics for MAssive and Networked DAta'',
    prot. 2012C4E3KT\_001, and NSERC.}}
\else
\title{Monotone Simultaneous \SL{Paths Embeddings}  in $\mathbb{R}^d$}
\fi
\author{David Bremner\inst{1} \and
Olivier Devillers\inst{2} \and
Marc Glisse\inst{3} \and
Sylvain Lazard\inst{2} \and
Giuseppe Liotta\inst{4} \and
Tamara Mchedlidze\inst{5} \and
Sue Whitesides\inst{6} \and
Stephen Wismath\inst{7}
}

\institute{
U. New Brunswick,Canada,~\email{bremner@unb.ca}\and
Inria, CNRS, U. Lorraine, France,~\email{olivier.devillers|sylvain.lazard@inria.fr}\and
Inria, Saclay, France,~\email{marc.glisse@inria.fr}\and
U. of Perugia, Italy,~\email{giuseppe.liotta@unipg.it}\and
KIT, Germany,~\email{mched@iti.uka.de}\and
U. of Victoria, Canada,~\email{sue@uvic.ca}\and
U. of Lethbridge, Canada,~\email{wismath@uleth.ca}}

\newcommand\blfootnote[1]{%
  \renewcommand\thefootnote{}\footnote{#1}%
  \addtocounter{footnote}{-1}%
}

\begin{document}

\maketitle

\begin{abstract}
  We study the following problem: Given $k$ paths that share the same vertex
  set, is there a simultaneous geometric embedding of these paths such
  that each individual drawing is monotone in some direction?  We  prove that
  for any dimension $d \geq 2$, there is a set of $d+1$ paths that does
  \emph{not} admit a monotone simultaneous geometric embedding.
\end{abstract}

\section{Introduction}\label{se:intro}

\ifshort
\blfootnote{Research supported in part by the MIUR
    project AMANDA ``Algorithmics for MAssive and Networked DAta'',
    prot. 2012C4E3KT\_001, and NSERC.}
\fi 
Monotone drawings and simultaneous embeddings are well studied topics in graph
drawing. Monotone drawings, introduced by Angelini et
al.~\cite{DBLP:journals/jgaa/AngeliniCBFP12}, are planar drawings of connected
graphs such that, for every pair of vertices,
there is a path between them that
monotonically increases with respect to some direction. Monotone drawings of
planar graphs have been studied both in the fixed and in the variable embedding
settings and both with straight-line edges and with bends allowed along
edges; recent papers on these topics
include~\cite{%
DBLP:journals/algorithmica/AngeliniDKMRSW15,%
felsner_et_al:LIPIcs:2016:5929,%
DBLP:journals/dmaa/Hossain015,%
kindermann2014monotone}.

The simultaneous (geometric) embedding problem was first described in a paper by Bra\ss{} et
al.~\cite{DBLP:journals/comgeo/BrassCDEEIKLM07}. The input is a set of planar graphs that share the
same labeled vertex set (but the set of edges differs from one graph to another); the output is a
mapping of the vertex set to a point set such that each graph admits a crossing-free drawing with
the given mapping. The simultaneous embedding problem has also been studied by restricting/relaxing
some geometric requirements; for example, while every pair of planar graphs sharing the same
labeled vertex set admits a simultaneous embedding where each edge has at most two bends (see,
e.g.,~\cite{DBLP:journals/jgaa/ErtenK05,DBLP:journals/cj/GiacomoDLMW15}), not even a tree and a path
always admit a geometric simultaneous embedding (such that the edges are straight-line
segments)~\cite{DBLP:journals/jgaa/AngeliniGKN12}). See the book chapter on simultaneous embeddings
by T. Bl\"asius et al.~\cite{blasius} for an extensive list of references on
the problem and its variants.

In this paper, we combine the two topics of simultaneous embeddings
and monotone drawings.  Namely, we are interested in computing
geometric simultaneous embeddings of paths such that each path is
monotone in some direction. Let $V ={ 1, 2, \ldots, n }$ be a labeled
set of vertices and let $\Pi =\{ \pi_1,\pi_2,\ldots, \pi_k \}$ be a set
of $k$ distinct paths each having the same set $V$ of vertices. We
want to compute a labeled set of points $P = \{ p_1, p_2,\ldots, p_n \}$
such that point $p_i$ represents vertex $i$ and for each path $\pi_i
\in \Pi$ ($1 \leq i \leq k$)
there exists some
direction for which the drawing of $\pi_i$
is monotone.

It is already known that any two paths on the same vertex set admit a monotone simultaneous
geometric embedding in 2D, while there exist three paths on the same vertex set for which a
simultaneous geometric embedding does not exist even if we drop the monotonicity requirement~\cite{DBLP:journals/comgeo/BrassCDEEIKLM07}.
 An example of three paths that do not have a monotone simultaneous
geometric embedding in 2D can also be derived from a paper of Asinowski on suballowable
sequences~\cite{DBLP:journals/dm/Asinowski08}.
%
%
On the other hand, it is immediate to see that in 3D
any number of paths sharing the same vertex set admits a simultaneous geometric embedding: Namely,
by suitably placing the points in generic position (no 4 coplanar), the complete graph has a
straight-line crossing-free drawing; however, the drawing of each path may not be monotone. This
motivates the following question: Given a set of paths sharing the same vertex set, does the set
admit a monotone simultaneous geometric embedding in $d$-dimensional space for $d \geq 3$?

Our main result is that for any dimension $d \geq 2$, there exists a
set of $d + 1$ paths that does not admit a monotone simultaneous
geometric embedding in $d$-dimensional space. Our proof exploits the
relationship between monotone simultaneous geometric embeddings in
$d$-dimensional space and their corresponding representation in the
dual space. Our approach extends to $d$ dimensions the primal-dual
technique described in a recent paper by Aichholzer et
al.~\cite{DBLP:journals/jgaa/AichholzerHLMPV15} on simultaneous
embeddings of upward planar digraphs in 2D.
\iflong We also
discuss how to test whether a given set of paths sharing the same
vertex set admits a monotone simultaneous geometric embedding in
2D.\syl{This is not done in 2D but in $d$D}}
\else
\syl{removed the sentence about testing the existence of a monotone simultaneous embedding (which is
not done in this short paper)}
\fi

\iflong
The rest of the paper is organized as follows. Preliminaries are in
Section~\ref{se:def}, the primal-dual approach to study monotone
simultaneous geometric embeddability of paths and our main result on
existence of non-embeddable instances of $d+1$ paths in $d$ dimensions
are described in Section~\ref{se:primal-dual}. Testing
paths for monotone simultaneous geometric embeddability in
2D
is studied in the full version of this paper.
\fi

\section{Definitions}\label{se:def}

Let $\dir$ be a vector in $\mathbb{R}^d$ and let $G$ be a directed
acyclic graph with vertex set $V$.
An embedding $\Gamma$
  of  the vertex set $V$ in $\mathbb{R}^d$ is called
\emph{$\dir$-monotone} for $G$ if the 
vectors
 in $\mathbb{R}^d$ corresponding to
oriented edges of $G$ 
have a positive scalar product with $\dir$.
\iflong 

\else \syl{joined the 2 paragraph (saves one line)}
\fi
\sloppy Let $\Dirs=\{\dir_1,\dots,\dir_k\}$ be a set of $k > 1$ vectors in $\mathbb{R}^d$
and let $\Graphs=\{G_1, G_2, \dots, G_k\}$ be a set of $k$ distinct
acyclic digraphs on the same vertex set $V$.
A \emph{$\Dirs$-monotone simultaneous embedding} of $\Graphs$ in
$\mathbb{R}^d$ is an embedding $\Gamma$ of $V$ that is
$\dir_i$-monotone for $G_i$ for any \iflong value of\fi
 $i$.
A \emph{monotone simultaneous embedding} of $\Graphs$ is a $\Dirs$-monotone simultaneous embedding for some set $\Dirs$ of vectors.

If a graph is a path on $n$ (labeled) vertices, it can be trivially identified
with a permutation of $[1,n]$.
We look at the
 monotone simultaneous embedding
problem in the dual space, by mapping points representing
vertices to hyperplanes in  $\mathbb{R}^d$. The dual formulation of
monotone simultaneous embeddings \syl{added an s to embedding} is as follows (the equivalence of
these formulations is shown in the next section).
Let $\Pi=\{\pi_1,\pi_2,\ldots,\pi_k\}$ be a set of $k$ permutations of $[1,n]$.
A \emph{parallel simultaneous embedding} of $\Pi$ in  $\mathbb{R}^d$ is a set of $n$
hyperplanes $H_1, H_2, \ldots , H_n$ and $k$ vertical lines
$L_1, L_2, \ldots, L_k$
 such that the set of $n$ points
$L_j\cap H_{\pi_j(1)}, 
\ldots,L_j\cap H_{\pi_j(n)}$
is linearly ordered from bottom to top along $L_j$,  for all $j$.



\section{The Dual Problem and Non-Existence Results}\label{se:primal-dual}

The first two lemmas give duality results between
monotone simultaneous embeddings
and parallel simultaneous embeddings. \syl{added 2 s}

\begin{lemma}
\label{lem:duality}
If a set of $k$ permutations of $[1,n]$ admits a parallel simultaneous
embedding in $d$ dimensions,
it also admits a monotone simultaneous embedding in $d$ dimensions.
\end{lemma}
\begin{proof}
Consider the following duality between points and hyperplanes,
where we denote by $H^\star$ the dual of a non-vertical hyperplane $H$:
\[ \textstyle  H: x_d = \pth{\sum_{i=1}^{d-1} \alpha_i x_i} - \alpha_0 ,
\qquad\qquad
 H^\star = (\alpha_1, \ldots, \alpha_{d-1}, \alpha_0). \]
  This duality maps
 parallel hyperplanes to points that are vertically aligned (and vice-versa).
Let $(H_i)_{1\le i\le n}$, $(L_j)_{1\le j\le k}$ be a parallel simultaneous embedding and refer
  to Fig.~\ref{fig:duality}
By definition, line $L_j$ crosses hyperplanes $H_1,\dots, H_n$ in the order $ H_{\pi_j(1)}, H_{\pi_j(2)},\ldots, H_{\pi_j(n)}$.
The intersection points $L_j\cap H_{\pi_j(1)},L_j\cap H_{\pi_j(2)},\ldots,L_j\cap H_{\pi_j(n)}$ are
collinear and therefore represent parallel hyperplanes in the dual plane. Consider the vector
  line $\dir_j$
 perpendicular to these hyperplanes and pointing downward.
This line crosses them in the order $(L_j\cap H_{\pi_j(1)})^\star, (L_j\cap H_{\pi_j(2)})^\star, \ldots,(L_j\cap H_{\pi_j(n)})^\star$. Since point $H_i^\star$ lies in hyperplane $(L_j\cap H_i)^\star$,
points ${H_i}^\star,1\leq i\leq n$, project on $\dir_j$
in
the order $H_{\pi_j(1)}^\star,H_{\pi_j(2)}^\star,\ldots, H_{\pi_j(n)}^\star$.
Therefore $(H^\star_i)_{1\le i\le n}$ is an embedding such that
path $\pi_j$ is $\dir_j$-monotone, for all $j$.\qed
\end{proof}

\begin{figure}[t]
\begin{center}
\includegraphics[width=0.65\textwidth,page=1]{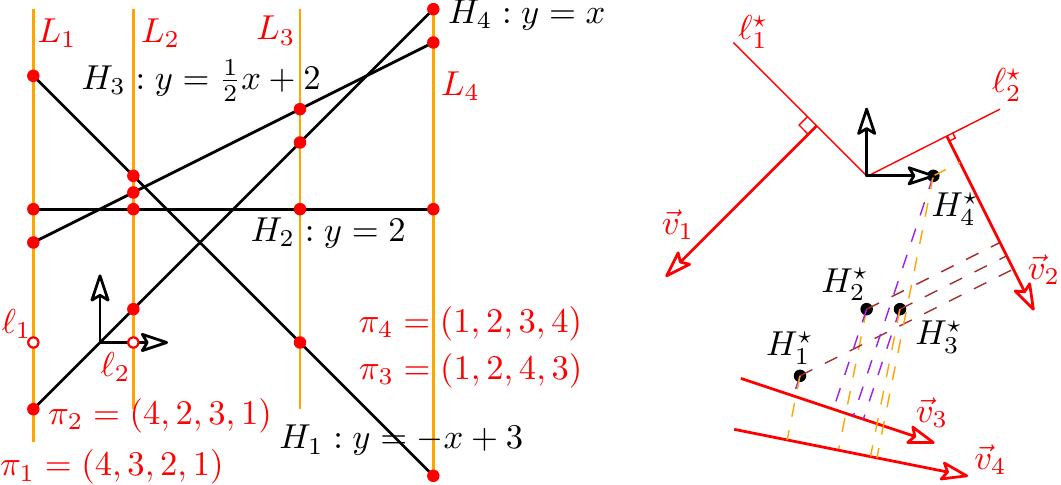}
\end{center}
\caption{Duality between monotone simultaneous embeddings and parallel
  simultaneous embeddings for $k=n=4$ and $d=2$. \iflong Notice that $L_1$ and $L_4$
  correspond to reverse permutations.\fi \label{fig:duality}}
\end{figure}

\syl{Removed in caption of Fig \ref{fig:duality} for length "Notice that $L_1$ and $L_4$
  correspond to reverse permutations}

\begin{lemma}
\label{lem:duality-reciprocal}
If a set $(\pi_j)_{1\le j\le k}$ of $k$ permutations of $[1,n]$ admits a monotone simultaneous
embedding in $d$ dimensions,
there is a set $(\pi'_j)_{1\le j\le k}$
that admits a parallel simultaneous embedding in $\R^d$ where,  for every $j$,
$\pi'_j$ is either equal to \SL{$\pi_j$ or to its reverse.}
\end{lemma}
\begin{proof}
As in the proof of Lemma~\ref{lem:duality},
we consider point-hyperplane duality. Let
$(p_i)_{1\le i\le n}$ be an embedding $\vec{v}_j$-monotone for
$\pi_j$, and $(p^\star_i)_{1\le i\le n}$ the corresponding set of dual hyperplanes.
Let $H_j$ be a hyperplane with normal vector $\vec{v}_j$, $1\leq j \leq n$.
Define  $L_j$ to be the vertical line through point $H_j^\star$.
By construction, the points $\pth{L_j\cap p^\star_{\pi_j(i)}}_{i}$
appear in order on $L_j$ for one of the two possible orientations of $L_j$.
In particular,  when $\dir_j$ points {downward}, $L_j$ lists the points
 ${L_j\cap p^\star_{\pi_j(i)}}$
from  bottom to top and vice versa.\qed
\end{proof}


We now prove results of existence and
 non-existence of parallel simultaneous embeddings\iflong for certain
configurations\fi, starting with 
a very simple result of existence.

\begin{proposition}
\label{lem:exist}
Any set of $d$ permutations on $n$ vertices  admits a monotone
simultaneous embedding
and a parallel
simultaneous embedding
in $d$ dimensions.
\end{proposition}
\begin{proof}
Choose $d$ points in general position in the hyperplane $x_d=0$
and draw a vertical line through
 each of these points.
For each  vertical line, choose a permutation and place on the line  $n$ points numbered according to the permutation.
Fit a hyperplane
through all the points with the same number.
By construction, this set of hyperplanes is a  parallel
simultaneous embedding.
Going to the dual, by Lemma~\ref{lem:duality}, gives a monotone
simultaneous embedding.
Alternatively, the monotone embedding can be seen directly by
considering the rank in the $i$-th permutation as the $i$-th
coordinate.\qed
\end{proof}
\iflong
It is interesting to contrast this construction with the difficulty
of realizing permutations as line transversals of \emph{disjoint}
convex sets; in~\cite{DBLP:journals/dcg/AsinowskiK05} the authors show
for any $k\geq d/2+1$ there exists a family of $k$ permutations not so
realizable in $\R^d$. In particular there exist $3$ permutations not
realizable as line transversals of disjoint convex sets in $\R^3$.
\else
\syl{Removed the discussion on  the difficulty
of realizing permutations as line transversals of \emph{disjoint}
convex sets; We can keep it : it fits in 6 pages with it but I find Figures 2 and 3 really small and
this allows to make them bigger.}
\fi
We now turn our attention to non-existence. For proving that there exists $k=d+1$ permutations that do not admit a parallel
  simultaneous embedding in $d$ dimensions, observe that we can consider any generic placement of the $d$
  first lines $L_j$ since all such placements are equivalent through affine transformations. We then
  construct permutations for $n$ big enough that cannot be realized with any placement of
  $L_{d+1}$. Similarly, constructing $k=d+1$ permutations that cannot be realized even up to
  inversion, yields the non-existence of a monotone simultaneous embedding in $d$ dimensions by Lemma~\ref{lem:duality-reciprocal}. 
We start with dimension~2, then move to dimension~3 and only then, generalize our results to
arbitrary dimension.
Observe that 2D results also follow from
 \cite[Lemma~1 \& Prop. 8]{DBLP:journals/dm/Asinowski08}, but we still present our proofs as a
 warm up for higher  dimensions.

\begin{lemma}
\label{lem:dual_view_2d}
There exists a set of $3$ permutations on $\{0,1,2\}$ that does
not admit a parallel simultaneous embedding in 2D.
\end{lemma}
\begin{proof}
Let $L_1$ and $L_2$ be two vertical lines,
$H_1$ and $H_2$ two other lines,
and let $\tau_1=(1,2)$ and $\tau_2=(2,1)$
be two permutations of $\{1,2\}$.
As in Fig.~\ref{fig:2d}-left, if $L_1$ is left of $L_2$ and
the intersections of $H_1$ and $H_2$ with $L_j$ are ordered according to $\tau_i$,
we can deduce that $H_1\cap H_2$ is between $L_1$ and $L_2$.
It follows that a vertical line crossing $H_1$ below $H_2$ is to the
left of that intersection point and thus to the left of $L_2$.
Similarly,  a vertical line crossing $H_1$ above $H_2$ is to the
right of $L_1$.
\begin{figure}[t]
\begin{center}
\includegraphics[width=0.75\textwidth,page=2]{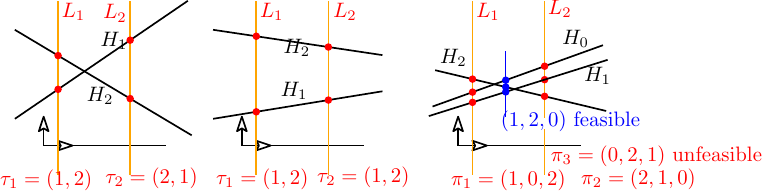}
\end{center}
\caption{Non-existence of two-dimensional parallel simultaneous embeddings.\label{fig:2d}}
\end{figure}
If we now consider $\tau_1=\tau_2=(1,2)$ we have that a vertical line
crossing $H_1$ above $H_2$ is not between $L_1$ and $L_2$
(Fig.~\ref{fig:2d}-center).
Consider now
$\pi_1=(1,0,2)$,
$\pi_2=(2,1,0)$ and
$\pi_3=(0,2,1)$.
Restricting the permutations to $\{1,2\}$ gives that $L_3$ must
be right of $L_1$,
restricting  to $\{0,2\}$ gives that $L_3$ must
be left of $L_2$, and
restricting  to $\{0,1\}$ gives that $L_3$ cannot be between $L_1$
and  $L_2$
(Fig.~\ref{fig:2d}-right).
We deduce that no placement for $L_3$ can realize $\pi_3$.
Notice that the reverse order $(1,2,0)$ can be realized and thus the
dual of this
construction is not a counterexample to simultaneous monotone embeddings.\qed \syl{added an s}
\end{proof}

\begin{lemma}
\label{lem:primal_view_2d}
There exists a set of $3$ permutations on $6$ vertices that does
not admit a monotone simultaneous embedding in 2D.
\end{lemma}
\begin{proof}
Let $\pi_1=(f,b,d,e,a,c)$,         
$\pi_2=(d,f,c,b,e,a)$, and  
$\pi_3=(f,a,d,c,e,b)$.         
%
%
The sub-permutations of $\pi_1, \pi_2$ and $\pi_3$ on $\{a,b,c\}$  are (by matching $(a,b,c)$ to
$(0,1,2)$) the 3 permutations  that  do not admit a parallel
simultaneous embedding in the proof of
Lemma~\ref{lem:dual_view_2d}.
The same is obtained by reversing only $\pi_1$ (resp. $\pi_2$,   $\pi_3$) and considering sub-permutations on $\{a,c,d\}$
(resp.  $\{d,b,e\}$,  $\{b,f,d\}$). Other possibilities are symmetric and  Lemma~\ref{lem:duality-reciprocal} yields the result.\qed
\end{proof}

\begin{lemma}
\label{lem:dual_view_3d}
There exists a set of $4$ permutations on $5$ vertices that does
not admit a parallel simultaneous embedding in 3D.
\end{lemma}
\begin{proof}
As in the proof of Lemma~\ref{lem:exist} we consider $3$ points $\ell_1,\ell_2,\ell_3$
in general position in the hyperplane $x_3=0$ and the $3$
vertical lines $L_1,L_2,L_3$ going through these points.
Let $L$ be a vertical line (candidate position for $L_4$)
and $\ell$ its  intersection with  $x_3=0$.
We consider the $3$ permutations
$\tau_1=(1,2,3)$,
$\tau_2=(2,3,1)$,
$\tau_3=(3,1,2)$
defining the vertical order of the intersections  of $L_1,L_2,L_3$ with
hyperplanes $(H_i)_{1\le i\le 3}$.
We denote by $h_{i,j}$ the projection of the line $H_i\cap H_j$, $1\leq i\neq j \leq 3$, onto the plane $x_3=0$.
Since the three planes $H_i$, $1\leq i \leq 3$
meet in one  point,
the lines $h_{1,2}$, $h_{2,3}$ and $h_{1,3}$ meet at the
projection of that point onto the plane $x_3=0$.

Refer to Fig.~\ref{fig:3d}.
For $L$ to cut $H_2$ below $H_1$,
$\ell$ must be in  the half-plane  limited by
$h_{1,2}$ and containing $\ell_2$,
and, similarly, for $L$ to cut $H_3$ below $H_2$, $\ell$ must be in the
half-plane limited by $h_{2,3}$ and containing $\ell_3$.
Thus, $\ell$ must be in a wedge with apex $h_{1,2}\cap h_{2,3}$
(Fig.~\ref{fig:3d}-left).
Since $h_{1,2}$ separates $\ell_2$ from $\ell_1$  and $\ell_3$, and
 $h_{2,3}$ separates $\ell_3$ from $\ell_1$ and $\ell_2$,
the union of all
 wedges, for all possible positions of $h_{1,2}$ and $h_{2,3}$, is the
union, $\mathcal R$, of triangle $\ell_1\ell_2\ell_3$ and the half-plane limited by $\ell_2\ell_3$ and not
  containing $\ell_1$
(Fig.~\ref{fig:3d}-center).
To summarize, if
$\tau_1=(1,2,3)$,
$\tau_2=(2,3,1)$,
$\tau_3=(3,1,2)$,
and $\tau_4=(3,2,1)$
then
$\ell_4$ (the intersection point of $L_4$ with the hyperplane $x_3=0$) must lie in this region $\mathcal R$.

Next, we build the permutations $\pi_1,\pi_2,\pi_3$ and $\pi_4$ by repeating this example as
follows: $\pi_1=(0,1,2,3,4)$, $\pi_2=(2,3,4,0,1)$, $\pi_3=(3,4,0,1,2)$, and $\pi_4=(1,3,2,0,4)$.
The restriction of these permutations to $\{0,2,3\}$ yields that $\ell_4$ must be in the triangle or
in the half-plane limited by $\ell_2\ell_3$ and not containing $\ell_1$. The restriction to
$\{1,2,3\}$ yields that $\ell_4$ must be in the triangle or  in the half-plane limited by
  $\ell_1\ell_3$ and not containing $\ell_2$.  The restriction to $\{0,2,4\}$ yields that $\ell_4$
must be in the triangle or in the half-plane limited by $\ell_1\ell_2$ and not containing $\ell_3$.
Finally, considering $\{0,1\}$ yields that $\ell_4$ must be outside the triangle (Fig.~\ref{fig:3d}-right).
\begin{figure}[t]
\begin{center}
\includegraphics[width=1\textwidth,page=3]{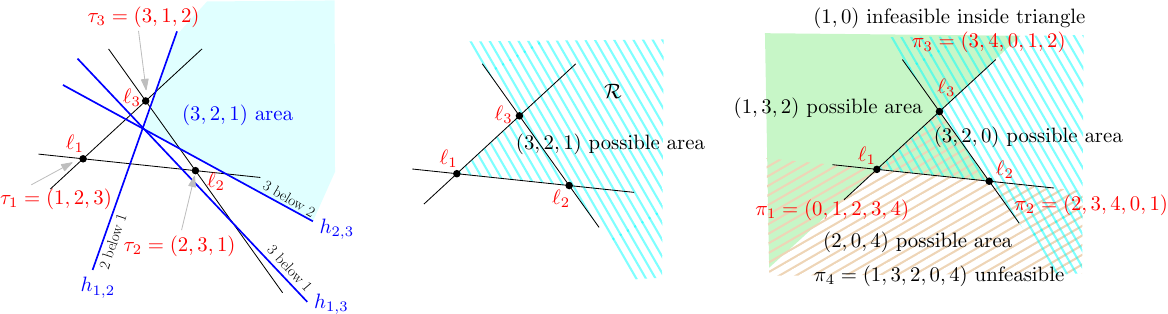}
\end{center}
\caption{Non-existence of \iflong three-dimensional \else 3D \fi parallel simultaneous
  embeddings for 5 vertices.\label{fig:3d}}
\end{figure}
Thus there is no possibility for placing $L_4$.
\qed\end{proof}

\begin{lemma}
\label{lem:primal_view_3d}
There exists a set of $4$ permutations on $40$ vertices that does
not admit a monotone simultaneous embedding in 3D.
\end{lemma}
{\em Sketch of proof.} 
The idea is to concatenate several versions of the counterexample of
the previous lemma to cover all
possibilities of reversing
permutations.
Note that the number of 40
 vertices is not tight.
\qed 

\begin{lemma}
\label{lem:dual_view}
There exists a set of $d+1$ permutations on  $3\cdot 2^d$ 
vertices   that   does
not admit a parallel simultaneous embedding in $d$ dimensions.
\end{lemma}
{\em Sketch of proof.} 
As in Lemma~\ref{lem:dual_view_3d}, the idea is to consider the simplex $(\ell_j)_{1\le j\le d}$
and to construct the permutations for the $L_i$
in order to prevent all possibilities for placing $\ell_{d+1}$.
\qed

To get a result in the dual, the difficulty is that we have to prevent
not only
some permutations but also their reverse versions.

\begin{theorem}
\label{th:primal_view}
There exists a set of $d+1$ permutations on $3\cdot 2^{2d}$ vertices that does
not admit a monotone simultaneous embedding in $d$ dimensions.
\end{theorem}
{\em Sketch of proof.} 
As for Lemma~\ref{lem:primal_view_3d} we concatenate several versions
of previous  counter-example to cover all possibilities of reversing permutations.
\qed
\ifshort
\fi



\bigskip
{\footnotesize
\noindent{\bf Acknowledgements.}
 This work  was initiated during the 15$^{th}$ {\em
 INRIA--McGill--Victoria Workshop on Computational Geometry} at the
 Bellairs Research Institute. The authors wish to thank all the
 participants for creating a pleasant and stimulating atmosphere.
}

\bibliographystyle{abbrvurl}
\bibliography{biblio}


\newpage
\appendix

\section{Proofs\label{se:proofs}}
\noindent\emph{Proof of  Lemma~\ref{lem:primal_view_3d}.}
We consider \\
$\pi_1$={\scriptsize$
  (0,1,2,3,4,\;10,11,12,13,14,\;20,21,22,23,24,\;30,31,32,33,34,\;40,41,42,43,44,\;\\
~\hfill 50,51,52,53,54,\; 60,61,62,63,64,\;70,71,72,73,74)$,}\\
$\pi_2$={\scriptsize$ (2,3,4,0,1,\;12,13,14,10,11,\;22,23,24,20,21,\;32,33,34,30,31,\;41,40,44,43,42,\;\\
~\hfill 51,50,54,53,52,\;61,60,64,63,62,\;71,70,74,73,72)$,}\\
$\pi_3$={\scriptsize$ (3,4,0,1,2,\;13,14,10,11,12,\;22,21,20,24,23,\;32,31,30,34,33,\;43,44,40,41,42,\;\\
~\hfill 53,54,50,51,52,\;62,61,60,64,63,\;72,71,70,74,73)$,} and\\
$\pi_4$={\scriptsize$(1,3,2,0,4,\;14,{10},12,13,11,\;21,23,22,20,24,\;34,{30},32,33,31,\;41,43,42,40,44,\;\\
~\hfill 54,{50},52,53,51,\;61,63,62,60,64,\;74,{70},72,73,71)$}\\
The idea is that we have eight groups of vertices.
Group $\{0,1,2,3,4\}$ restricts exactly to the example of Lemma~\ref{lem:dual_view_3d}
 and prevents going from primal to dual without reversing any
permutations.
Group $\{10,11,12,13,14\}$ prevents going from primal to dual
reversing exactly $\pi_4$.
The other groups prevent all combinations of reversals that leave the first permutation fixed.
In this example we prefer the simplicity of proof to optimizing
the number of vertices. Counterexamples with less vertices can be
easily obtained by sharing vertices between the different groups.
\qed

\medskip\noindent\emph{Proof of  Lemma~\ref{lem:dual_view}.}
As in previous lemma, we generalize Lemma~\ref{lem:dual_view_2d}
without trying to optimize the number of vertices in the permutations.
We consider $d$ points $(\ell_j)_{1\le j\le d}$
in general position in the hyperplane $x_d=0$ and the $d$
vertical lines $(L_j)_{1\le j\le d}$ going through these points.
Let $L_{d+1}$ be a (variable) vertical line and $\ell_{d+1}$ its  intersection with
$x_d=0$.
In a similar manner as in two dimensions consider
$\tau_1=(1,0,2)$,
$\tau_2=(2,1,0)$, and
$\tau_3=(0,2,1)$
and
$\Pi_1\subset \{i\,|\, 1\leq i\leq d\}$,
$\Pi_2= \{i\,|\,1\leq i\leq d\}\setminus \Pi_1$, and
$\Pi_3= \{{d+1}\}$;
then assume that $\tau_i$  is the order of \SL{hyperplanes} $H_0, H_1, H_2$ along \SL{$L_k$} 
for \SL{any} $k \in  \Pi_i$.
\SL{In other words, above $\ell_k$, we have for instance $H_2$ above
$H_1$ for $k\in \Pi_1$ and the converse for $k\in\Pi_2\cup\Pi_3$.}

In projection, this means that $h_{1,2}=H_1\cap H_2$ separates
$\pth{\ell_i}_{i\in\Pi_1}$ from $\pth{\ell_i}_{i\in\Pi_2}$ and
that $\ell_{d+1}$ is on the side of  $\pth{\ell_i}_{i\in\Pi_2}$.
Thus, $\ell_{d+1}$ must be in the pink hatched part in
Fig.~\ref{fig:dd}.
\begin{figure}[t]
\begin{center}
\includegraphics[width=0.5\textwidth,page=4]{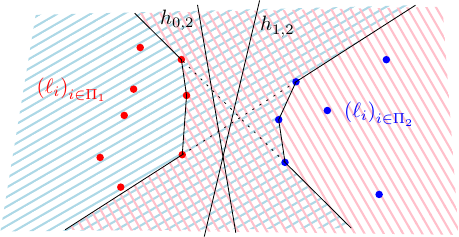}
\end{center}
\caption{Non-existence of a $d$-dimensional parallel simultaneous
  embedding.
\label{fig:dd}}
\end{figure}
Considering $h_{0,2}$ yields similarly that $\ell_{d+1}$
must be in the blue
hatched part, and consequently, there is a hyperplane through $\ell_{d+1}$
that separates $\pth{\ell_i}_{i\in\Pi_1}$ from
$\pth{\ell_i}_{i\in\Pi_2}$.

Now we construct $\pi_1, \ldots, \pi_{d+1}$ by concatenating
one copy of  $\tau_1$, $\tau_2$, and $\tau_3$ with three new vertices
for each possible partition of $\{i\,|\, 1\leq i\leq d\}$ in $\Pi_1$ and
$\Pi_2$.
\SL{For any such partition,  there is a hyperplane through $\ell_{d+1}$
that separates $\pth{\ell_i}_{i\in\Pi_1}$ from
$\pth{\ell_i}_{i\in\Pi_2}$.}
\SL{
Points $(\ell_j)_{1\le j\le d+1}$ can be seen in $\R^{d-1}$ (since $x_d=0$) and
considering the partition with $\Pi_1=\emptyset$ yields that  there is a hyperplane (in $\R^{d-1}$)
through $\ell_{d+1}$ with all  $(\ell_j)_{1\le j\le d}$ on one side. In other words, there is  a
hyperplane  (in $\R^{d-1}$) separating  $\ell_{d+1}$ from $(\ell_j)_{1\le j\le d}$. Projecting
$(\ell_j)_{1\le j\le d}$ onto that plane (with a central projection with center $\ell_{d+1}$) yields
$d$ points in $\R^{d-2}$, which can be partitioned in two sets,
whose convex hulls intersect by
Radon's theorem~\cite{Radon1921}. For this partition, there is no hyperplane through $\ell_{d+1}$
that separates $\pth{\ell_i}_{i\in\Pi_1}$ from
$\pth{\ell_i}_{i\in\Pi_2}$, which is a contradiction. Hence, these $d+1$ permutations on $3\cdot 2^d$ vertices prevent all
placements for $\ell_{d+1}$, which concludes the proof. (Note however that this number of vertices is clearly non-optimal.)
\qed
}


\medskip\noindent\emph{Proof of  Theorem~\ref{th:primal_view}.}
\sloppy A counterexample of $d+1$ permutations $\pth{\pi_j}_{1\le j\le d}$ with no monotone simultaneous
embedding must be a counterexample of $d+1$ permutations with no
parallel  simultaneous embedding for any set of permutations
obtained from  $\pth{\pi_j}_{1\le j\le d}$ by reversing some of these
permutations.
Since there are $2^d$ ways of choosing which permutations are reversed,
we can concatenate $2^d$  images of counterexamples from
Lemma~\ref{lem:dual_view} by reversing some permutations
so that the situation of Lemma~\ref{lem:dual_view} appears whatever
choice of reversing is done.
\qed

\end{document}